\documentclass[review]{elsarticle}
\usepackage{hyperref}
\usepackage{amsmath}
\usepackage{amssymb}
\usepackage{amsfonts}
\usepackage{amsthm}

\newcommand{\NN}[0]{\mathbb{N}}
\newcommand{\EE}[0]{\mathbb{E}}
\newcommand{\PP}[0]{\mathbb{P}}
\newcommand{\ZZ}[0]{\mathbb{Z}}

\newcommand{\bu}[0]{\mathbf{u}}
\newcommand{\bv}{\mathbf{v}}
\newcommand{\bx}[0]{\mathbf{x}}

\newcommand{\bs}[0]{\mathbf{s}}
\newcommand{\bS}[0]{\mathbf{S}}

\newcommand{\bi}[0]{\mathbf{i}}
\newcommand{\br}[0]{\mathbf{r}}
\newcommand{\bI}[0]{\mathbf{I}}

\newcommand{\bR}[0]{\mathbf{R}}

\newcommand{\Exp}[0]{\text{Exp}}

\newcommand{\limt}[0]{\lim_{t \rightarrow \infty}}

\newcommand{\cF}[0]{\mathcal{F}}

\newcommand{\bX}[0]{\mathbf{X}}
\newcommand{\bY}[0]{\mathbf{Y}}

\newcommand{\KK}[0]{\mathbb{K}}
\newcommand{\cS}[0]{\mathcal{S}}
\newcommand{\cI}[0]{\mathcal{I}}
\newcommand{\cR}[0]{\mathcal{R}}

\newcommand{\bzero}{\mathbf{0}}

\newcommand{\rs}{\mathrm{s}}
\newcommand{\ri}{\mathrm{i}}
\newcommand{\rr}{\mathrm{r}}

\newtheorem{mythm}{Theorem}

\newtheorem{mydef}[mythm]{Definition}
\newtheorem{example}[mythm]{Example}

\journal{Mathematical Biosciences}

\begin{document}

\begin{frontmatter}

\title{An epidemic model for an evolving pathogen with strain-dependant immunity}


\author[mymainaddress,mysecondaryaddress]{Adam Griffin\corref{mycorrespondingauthor}}
\cortext[mycorrespondingauthor]{Corresponding author}
\ead{adagri@ceh.ac.uk}

\author[mymainaddress]{Gareth O. Roberts}
\ead{g.o.roberts@warwick.ac.uk}

\author[mymainaddress]{Simon E.F. Spencer}
\ead{s.e.f.spencer@warwick.ac.uk}

\address[mymainaddress]{Department of Statistics, University of Warwick, Coventry, UK}
\address[mysecondaryaddress]{Centre for Ecology \& Hydrology, Benson Lane, Wallingford, Oxfordshire, OX10 8BB, UK}

\begin{abstract}
Between pandemics, the influenza virus exhibits periods of incremental evolution via a process known as antigenic drift. This process gives rise to a sequence of strains of the pathogen that are continuously replaced by newer strains, preventing a build up of immunity in the host population. In this paper, a parsimonious epidemic model is defined that attempts to capture the dynamics of evolving strains within a host population. The `evolving strains' epidemic model has many properties that lie in-between the Susceptible-Infected-Susceptible and the Susceptible-Infected-Removed epidemic models, due to the fact that individuals can only be infected by each strain once, but remain susceptible to reinfection by newly emerged strains. Coupling results are used to identify key properties, such as the time to extinction. A range of reproduction numbers are explored to characterise the model, including a novel quasi-stationary reproduction number that can be used to describe the re-emergence of the pathogen into a population with `average' levels of strain immunity, analogous to the beginning of the winter peak in influenza. Finally the quasi-stationary distribution of the evolving strains model is explored via simulation.
\end{abstract}


\end{frontmatter}


\section{Introduction}

Epidemic models have become important tools for understanding, predicting and developing mitigation strategies for public health planners dealing with infectious diseases. Recent advances in genetic epidemiology have greatly accelerated our understanding of the complex interactions between host immunity and pathogen evolution, and emphasised the important role that pathogen evolution can have on the dynamics of infection. However, it remains extremely challenging to combine together these two interacting processes within the same mathematical framework \cite{Kucharski2016}. In this paper we develop a parsimonious epidemic model that describes the transmission dynamics of a multi-strain pathogen with evolutionary dynamics similar to the influenza A virus evolving via antigenic drift. 

Multi-strain models have become increasingly popular due to the rise in availability of pathogen genetic analyses. Many models have been based on ordinary differential equations (ODE), despite the fact that stochastic effects play an important role in mutation \citep[see][for a review]{Kucharski2016}. Bichara \emph{et al.} \cite{Bichara2014} develop an epidemic model with competition between finitely many pathogen strains, and include vertical transmission and immunity from maternal antibodies in the infection dynamics. Meehan \emph{et al.} \cite{Meehan2018} analyse multi-strain epidemic models with mutation between strains within an ODE framework. However since their focus is on drug-resistance, they do not consider the effect of immunity. In the multi-strain models discussed in Gog \emph{et al.} \cite{gog2002}, there is assumed to be a finite number of possible strains, and each individual may be infected with one or more of such strains. Evolution was been modelled by a random jump process on a finite strain space using a nearest neighbour jump process. Models involving a countable number of infectious statuses have been discussed in the past \cite{moy1967}, but these typically only use the previously mentioned nearest-neighbour evolution. In \cite{moy1967} this is expressed as a model for parasitic infections where the ``type'' of an individual is defined by the quantity of parasites in a host.  Despite the many modelling papers on multi-strain epidemics, the methodology required to fit these models to data is only just emerging \citep{Touloupou2019}.

Between pandemics, the 4 main sub-types of the influenza virus evolve according to a process called antigenic drift \cite{brown2009}. Antigenic drift arises due to the fact that infection with a particular strain of influenza provides the host with a long-lasting immunity to future infection by the same strain. Once immunity to a particular strain has built up in the population, there is a selection advantage to strains that do not elicit the same immune response. To capture within a mathematical model the complex processes driving the evolution of the influenza virus is extremely challenging due to the interactions between host immunity and viral evolution \cite{bedford2014}. Nonetheless, simple models can give rise to surprisingly complex dynamics \cite{roberts2019,bedford2012}. The H3N2 subtype of influenza A, in particular, exhibits a narrow spread in its evolutionary tree, with all strains a short genetic distance from a single branch \cite{fitch1997,bedford2015}. Each strain persists for a relatively short amount of time before being replaced.   

In this paper we define a novel epidemic model with countably infinite, evolving strains that sits between the traditional susceptible-infected-susceptible (SIS) and susceptible-infected-removed (SIR) epidemic models, in that each individual may be infected many times with the pathogen, but only once by a strain. The model is designed to reflect the linear pattern of evolution observed in pathogens undergoing antigenic drift, such as seasonal influenza. By introducing an equivalence relation on the state space, we are able to describe the equilibrium behaviour of the model prior to elimination of the pathogen. In Section \ref{sec: limit}, coupling arguments are used to make precise the relationship between our new model and the traditional SIS and SIR models and to explore the large population limit. In Section \ref{sec: repro} we discuss three reproduction numbers for the novel model. Finally in Section \ref{sec: sim study} we explore simulations from the quasi-equilibrium distributions.

\section{Definition of the model}\label{sec: models}

\subsection{SIRS with evolving strains (E-SIRS)}\label{subsec: sirs}
Consider a closed population of $N$ individuals which are classified as \emph{susceptible}, \emph{infective} or \emph{removed}. For each time $t$, we denote the number of susceptibles by $\cS(t)$, the number of infectives by $\cI(t)$, and the number of removed individuals by $\cR(t)$. An infective remains in this class for a random period of time known as their infectious period, after which they become removed. Similarly, removed individuals become susceptible again after their immune period, during which they cannot be infected by any strain (even a new one).  We assume the durations of infectious periods are i.i.d. draws from $L_I \sim \Exp(\gamma)$ and immune periods are i.i.d draws from $L_R \sim \Exp(\delta)$.

To capture dynamics of competing and evolving strains, every individual has a strain index $k \in \ZZ$ which denotes the most recent strain with which an individual was (or is currently) infected. We denote the number of susceptibles, infectives and removed individuals respectively with strain index $k$ by $\cS_k(t)$, $\cI_k(t)$ and $\cR_k(t)$. Finally we denote by $K^*(t)$ the largest strain index observed up to time $t$ and use $K^* := K^*(t^-)$ where the time is clear from context.

As in the standard SIRS model \cite{nasell2002}, we assume homogeneous mixing of individuals, and so each pair of individuals makes contact at the points of a Poisson process with rate $\frac{\beta}{N} > 0$. New strains are introduced into the population in the following way. Each time an infective makes contact with a susceptible individual, we assume that with some probability $\theta \in [0,1]$, there is a successful infection of the susceptible with a previously unseen strain, which is given strain index $K^*(t^-)+1$. With probability $1-\theta$, the original strain in the infective attempts to infect the susceptible; the success of this infection depends on the strain index of the susceptible. For simplicity, we assume that immunity is cumulative: a susceptible with strain index $k$ is immune to all strains with index $j \leq k$. Removed and susceptible individuals retain the strain index of the strain they have most recently recovered from. All contact processes, mutation events, infectious periods and immune periods are assumed to be independent from each other.
	
To summarise, the epidemic proceeds according to the following events. 

\begin{itemize}
	\item \emph{Infection without mutation:} 
	
		$(\cS_j(t), \cI_k(t)) \mapsto (\cS_j(t) -1, \cI_k(t) +1)$

	for all $j < k$, with rate $\beta(1-\theta) N^{-1} \cS_j(t)\cI_k(t)$.
	\item \emph{Infection with mutation:}
	
		$(\cS_j(t), \cI_{K^*+1}(t)=0) \mapsto (\cS_j(t)-1, \cI_{K^*+1}(t)=1)$
	for $j \in \ZZ$ with rate $\beta \theta N^{-1} \cS_j(t) \cI(t)$.
	\item \emph{Recovery:}
	
		$(\cI_k(t), \cR_k(t)) \mapsto (\cI_k(t) - 1, \cR_k(t) +1)$
	for $k \in \ZZ$ with rate $\gamma \cI_k(t)$.
	\item \emph{Loss of global immunity:}

		$(\cR_k(t), \cS_k(t)) \mapsto (\cR_k(t) -1, \cS_k(t) + 1)$
	for $k \in \ZZ$ with rate $\delta \cR_k(t)$.
\end{itemize}

The state space of the E-SIRS model is given by $\Omega'= \{ (\bs, \bi, \br): \sum_{k \in \ZZ} (\rs_k + \ri_k + \rr_k) = N \}$, with $\bs, \bi, \br$ being infinite sequences taking values in $\{0,1,\dots,N\}$. A natural initial condition might be $\cI_1(0) = 1, \cI_k(0) = 0$ for $k \neq 1$, $\cS_0(0) = N-1$, $\cS_k(0) = 0$ for $k \neq 0$, and $\cR_k(0)=0$ for all $k$. This equates to a single currently infective individual infected with a strain to which all other individuals are susceptible, and to which no-one is currently recovering.

\subsection{SIS with evolving strains (E-SIS)}\label{subsec: sis}
Consider a second model where
following an infectious period, 
an individual becomes immediately susceptible, corresponding to the E-SIRS model where $\delta = \infty$.
In this model there are no periods of immunity and so recovery events generate susceptibles:

 \begin{itemize}
 	\item \emph{Recovery:}
 	
 		$(\cI_k(t), \cS_{k}(t)) \mapsto (\cI_k(t) - 1, \cS_{k}(t) +1)$
 	for $k \in \ZZ$ with rate $\gamma \cI_k(t)$.
 \end{itemize}

plus infection transitions as above.
The E-SIS model evolves over the subspace $\{ (\bs, \bi): \sum_{k \in \ZZ} (\rs_k + \ri_k) = N  \}$. We will refer to both spaces by $\Omega'$, the meaning will always be clear from context.

\subsection{Link to single-strain models}\label{subsec: no evo}
Consider the E-SIS model with $\theta = 1$. All contacts are mutation contacts and hence successful, and so $(\cS(t),\cI(t))$, the total numbers of susceptibles and infectives, follow a traditional single-strain SIS model as defined in \cite{Andersson2000}. We can also perform a similar identification between $(\cS(t),\cI(t),\cR(t))$, the number of susceptibles, infectives and immune individuals in the E-SIRS model and the single-strain SIRS model as defined in \cite{nasell2002}.

On the other hand, consider the E-SIS model with $\theta = 0$. Since no contacts are mutations, no individual can be infected more than once. If the population starts with strain index $0$ except for the initial infectives with strain $1$, $(\cS_{0}(t), \cI_{1}(t), \cS_{1}(t))$ behaves as a traditional single-strain SIR model as defined in \cite{Andersson2000}. 
%
%

\section{Equivalence relation}\label{sec: equiv}
We wish to study the long-term average behaviour of characteristics such as the levels of immunity and pathogen diversity, however the constant emergence and extinction of strains means that the evolving epidemic process has no steady-state. To counter this we introduce an equivalence relation to fix the process against the most recently emerged strain. 

\begin{mydef}
	The \emph{active strain set} of a state $(\bs, \bi, \br) \in \Omega'$ is given by $\KK = \{k \in \ZZ : \ri_k > 0\}$. Let elements of this set be indexed from 1 to $|\KK|$ in ascending order, so for $k_a, k_b \in \KK$, we have $k_a < k_b$ whenever $a < b$.
\end{mydef}
\begin{mydef}\label{def: equiv}
	Two states $(\bs,\bi,\br)$ and $(\bs', \bi', \br') \in \Omega'$ are equivalent if and only if the following conditions hold.
	\begin{enumerate}
		\item The total numbers of susceptibles, infectives and removed individuals are equal: $|\bs| = \sum_{k\in \ZZ} \rs_k = \sum_{k\in\ZZ} \rs'_k = |\bs'|$, and similarly $|\bi| = |\bi'|$ and $|\br| = |\br'|$.
		\item The numbers of active strains are equal: $|\KK| = |\KK'|$.

		\item Each active strain has the same number of infectives: $\ri_{k_a} = \ri'_{k'_a}$, for $a = 1, \dots, |\KK|$.
		\item The numbers of individuals that are susceptible to the $a$th active strain are equal: $\sum_{k < k_a} \rs_k = \sum_{k < k'_a} \rs'_k$, for $a = 1, \dots, |\KK|$. 
		\item The numbers of removed individuals that will become susceptible to the $a$th active strain are equal: $\sum_{k < k_a} \rr_k = \sum_{k < k'_a} \rr'_k$ for $a = 1, \dots, |\KK|$.
	\end{enumerate}
\end{mydef}

In order to easily refer to the equivalence classes, we define the following representative of each equivalence class.
\begin{mydef}\label{def: rep}
	The \emph{representative} of the equivalence class containing $(\bs, \bi, \br)$, denoted $(\bs^*, \bi^*, \br^*)$ is defined as follows. If $\KK \neq \emptyset$, denote the active strains for the representative by $\KK^* = \{ 1 - |\KK|, \dots, 0 \}$. Let $\phi: \KK \rightarrow \KK^*$ be a bijection defined by $\phi(k_a) = a - |\KK|$ for $a = 1\, \dots, |\KK|$. Then the representative $(\bs^*, \bi^*, \br^*)$ is given by:

\begin{align*}
		\ri^*_k &= \begin{cases}
						\ri_{\phi^{-1}(k)} & \text{for } k \in \KK^*, \\
						0 & \text{otherwise.}
					\end{cases}\\
		\rs^*_k &= \begin{cases} 
\sum_{j = \phi^{-1}(k)}^{\phi^{-1}(k+1) - 1} \rs_j &\text{for } k\in \{ 1-|\KK|,\dots, - 1 \},\\
\sum_{j = \phi^{-1}(0)}^\infty \rs_j &\text{for } k = 0,\\
\sum_{j = -\infty}^{\phi^{-1}(1 - |\KK|) - 1} \rs_j &\text{for } k=-|\KK|,\\
0 & otherwise.  
\end{cases} \\
\rr^*_k &= \begin{cases} 
\sum_{j = \phi^{-1}(k)}^{\phi^{-1}(k+1) - 1} \rr_j &\text{for } k\in \{ 1-|\KK|,\dots, - 1 \},\\
\sum_{j = \phi^{-1}(0)}^\infty \rr_j &\text{for } k = 0,\\
\sum_{j = -\infty}^{\phi^{-1}(1 - |\KK|) - 1} \rr_j &\text{for } k=-|\KK|,\\
0 & otherwise.  
\end{cases}
\end{align*}
If $\KK = \emptyset$ then $\ri^*_k = 0$ for all $k \in \ZZ$ and $\rs^*_0 = \sum_{j \in \ZZ} \rs_j$ and $\rs^*_k = 0$ for $k \neq 0$, and similarly for $\br$. Let $\{\bzero\}$ denote the set of all these absorbing states.
\end{mydef}
In the rest of this paper, the process of representatives on the space of equivalence classes (states described with starred states as in Definition \ref{def: rep}) will be referred to as the normalised process, and will be denoted by $(\bS^*, \bI^*, \bR^*)$ or $(\bS^*, \bI^*)$ as appropriate.
Definitions \ref{def: equiv} and \ref{def: rep} remove all strains with no infective individuals, and give index $0$ to the most recent strain to have emerged and have infectives. All susceptibles and removed individuals are given the strain index one less than the nearest infective above them in strain order. Any individuals immune to all existing strains are given strain $0$, as though they just recovered from the most recently emerged strain.

\begin{example}
	Consider the state $(\bs, \bi, \br) \in \Omega'$ 
	given by
	\begin{align*}
		(\rs_1, \dots, \rs_7) &= (0,0,1,1,0,0,1) \\
		(\ri_1, \dots, \ri_7) &= (0,1,0,0,0,1,0) \\
		(\rr_1, \dots, \rr_7) &= (1,0,0,1,1,0,0)
	\end{align*}
	where all remaining terms of $\bs$, $\bi$, $\br$ are zero. Strains 2 and 6 are active so $\KK = \{2, 6\} \Rightarrow \KK^* = \{-1,0\}$
	and the representative under the equivalence relation is
	given by
	\begin{align*}
		(\rs^*_{-2}, \dots, \rs^*_0) &= (0,2,1), \\
		(\ri^*_{-2}, \dots, \ri^*_0) &= (0,1,1), \\
		(\rr^*_{-2}, \dots, \rr^*_0) &= (1,2,0).
	\end{align*}
\end{example}

\paragraph{Notation}\label{para: notation}
Recall that without the equivalence relation, the state space of the epidemic process was $\Omega'$. We denote the state space of the normalised process over the set of equivalence class representatives by $\Omega$. 

We will use $\bx = (\bs, \bi, \br) \in \Omega$ with, for example, $\bs = (\rs_k)_k \in \{0,\dots,N\}^{\ZZ}$ to denote a typical element of the state space. We will also use, for example, $|\bs| = \sum_{k \in \ZZ} \rs_k$ to denote the total number of susceptibles. A random variable written in calligraphic type, e.g.\ $\cR_k(t)$, refers to a process without the equivalence relation. The corresponding variable written in roman type, e.g.\ $R_k(t)$, refers to the normalised process evolving over representatives from the equivalence classes.


\section{Quasi-stationarity and absorbing states}\label{sec: qsd evo}

Like many infectious disease models, the E-SIRS and E-SIS models defined in Section \ref{sec: models} have an absorbing class of states that corresponds to the population containing no infected individuals, $\bi=\bzero$. For finite population models, the absorbing state is reached with certainty in finite time, and so the limiting distribution is degenerate with no mass in non-absorbing states. However, like the single strain SIS model, these processes may not go extinct for a long time (individuals  can be reinfected) and the transient quasi-stable behaviour is of interest. The quasi-stationary distribution and limiting conditional distribution conditioned on the epidemic not going extinct, represent the long-term average behaviour for an endemic disease. 

\subsection{Properties of quasi-stationary distributions for epidemics}\label{subsec: qsd}
In this section and the rest of this paper, $\PP_{\bu}[A] = \PP[A | X(0) \sim \bu]$.
\begin{mydef}
  Let $X = (X(t))_{t \geq 0}$ be a Markov process on a countable state space $\Omega$ with absorbing state $0$ from which it cannot escape. Then a distribution $\bu$ on $\Omega \setminus \{0\}$ is a \emph{quasi-stationary distribution} (QSD) if $\PP_{\bu}[X(t) \in A | X(t) \neq 0] = \bu(A)$ for all $t \geq 0$.
  
  Given initial condition $\bv$ on $S=\Omega\setminus\{0\}$, $\bu$ is a $\bv$-\emph{limiting conditional distribution} (LCD) if $\limt \PP_{\bv}[X(t) \in A | X(t) \neq 0] = \bu(A)$. Note that, for processes where $S$ is a single communicating class, every QSD $\bu$ is a $\bu$-LCD and every LCD is a QSD.
\end{mydef}

Related to the QSD on irreducible state spaces is the notion of the decay parameter which describes the rate of decay of the transition probabilities.
\begin{mydef}
	Let $X=(X(t))_{t \geq 0}$ be an irreducible Markov process on a countable state space $\Omega$ with absorbing state $0$. Let $\bu$ be a QSD associated to $X$. Then the \emph{decay parameter} $\alpha$ is given by
	\begin{align*}
	\alpha = \inf \{ a \geq 0 : P_{ij}(t) = o(e^{-at}) \}.
	\end{align*}
    for $i, j \in \Omega \setminus \{0\}$
	The \emph{absorption parameter} $\alpha_0$ is given by
	\begin{align*}
	\alpha_0 = \inf \{a \geq 0 : \int_0^\infty \PP_i[T > t] e^{at} dt = \infty \},
	\end{align*}
    for $i \in \Omega \setminus \{0\}$
where $T$ is the extinction time of $X$ starting from state $i$. Note that for irreducible processes, $\alpha$ is independent of $i,j$ and $\alpha_0$ is independent of $i$.
\end{mydef}
According to Theorem 6 of \cite{vanDoorn2013}, a necessary condition for the existence of a QSD is that $\alpha > 0$.
\begin{mythm}\label{qsdepidemic}
  Conditional on non-absorption, the following hold.
  \begin{enumerate}
    \item The QSD for the number of infectives in the single-strain SIS model exists uniquely and gives weight to all states $\{1, \dots, N\}$.
    \item The QSD for the number of susceptibles and infectives in the single-strain SIR model exists uniquely and gives full weight to the state $\{(S,I)=(0,1)\}$.
    \item The QSD for the number of susceptibles and infectives in the single-strain SIRS model exists uniquely, and gives weight to all non-absorbing states.
  \end{enumerate}

\end{mythm}

\begin{proof}
Theorem 1 of \cite{vanDoorn2008} states that QSDs exist and are unique on finite irreducible state spaces, and so there is a unique QSD for the SIS model and for the SIRS model conditional on $\{I > 0\}$, and non-zero weight is given to all non-absorbing states. For reducible processes, Theorem 8 of \cite{vanDoorn2008} states that QSDs will give full weight to the communicating class with the longest expect time to leave and any states accessible from this ``slowest'' communicating class. This characterises the QSD for the SIR model.
\end{proof}

Further work on characterising the QSD for the standard SIS model can be seen in \cite{Nasell1999,nasell2011,clancy2003} making use of recurrent processes and normal approximations.

\subsection{Existence and uniqueness}\label{subsec: e and u}
Here we will summarise the existence and uniqueness results for the E-SIS and E-SIRS processes. 
\begin{mythm}\label{thm: eu SIStr}
	Let the E-SIS model be defined as in Subsection \ref{subsec: sis} with parameters $\beta, \gamma > 0$ and $\theta \in (0,1]$. Then for the normalised process, conditional on the events $\{I(t) > 0\}$, there exists a unique QSD which equals the unique LCD of the process and gives weight to all non-absorbing (i.e.\ transient) states, $\{(\bs,\bi)\in\Omega: |\bi| > 0\}$. If $\theta = 0$ and the process begins with a single infective, then there exists a unique LCD which gives full weight to the state with a single infective with strain index $0$, and $N-1$ susceptibles with strain index 0.
\end{mythm}
\begin{proof}
	For $\theta \in (0,1]$, we obtain existence and uniqueness by proving that $S=\Omega\setminus\{\bzero\}$ is a single finite communicating class, which immediately gives existence, uniqueness and equality of the QSD and LCD by Theorem 3 from \cite{vanDoorn2013}. 
	
	Under the equivalence relation, there can be at most $N$ different strain indices present in the population. This implies that every individual must appear in one of the states $\rs_{1-N},\dots \rs_0$ or $\ri_{1-N},\dots, \ri_0$. Therefore we can bound above the size of $S$, the set of transient (i.e.\ non-absorbing) states, by $(2N)^N$. 
	
	One can see that the transient states form a single communicating class by noting that one can get from a single infective of strain index $0$ with $N-1$ susceptibles of index 0 to any other state through infections (mutation and non-mutation) and recoveries. If all individuals are infected and then all but one recovers, then the process returns to the single infective case mentioned above.
	
	For $\theta = 0$ we consider the E-SIS model starting with a single infective of strain $0$ and susceptibles of strain index $-1$. If $\theta = 0$, then mutation is impossible. As a result, once an individual has become infected and recovered, they join the $s_0$ class and cannot be reinfected. In this way $\{S_0(t)\}$ behaves identically to the $\{R(t)\}$ class in the SIR model, and we identify the two models in this way. Point 2 in Theorem \ref{qsdepidemic} then gives the required LCD.
\end{proof}


\begin{mythm}\label{thm: eu SIRStr}
	Let the E-SIRS model be defined as in Subsection \ref{subsec: sirs} with parameters $\beta, \gamma, \delta > 0$ and $\theta \in (0,1]$. Then, conditional on having at least one infective, there exists a unique QSD. If $\theta = 0$ and there is initially one infective, a QSD still exists and gives full weight to the state with one infective with strain index $0$ and $N-1$ susceptibles with strain index $0$.
\end{mythm}
\begin{proof}
	For $\theta \in (0,1]$, one follows the same argument as in Theorem \ref{thm: eu SIStr}, this time bounding the size of the state space by $(3N)^N$, since individuals may also reside in classes $\rr_{1-N}, \dots, \rr_0$. The fact that the transient states form a single communicating class also follows as in Theorem \ref{thm: eu SIStr}.
	For $\theta = 0$, we see that each state that can be reached is a transient communicating class; there is no way to return to a state once left. As such, we need to consider the decay parameter on leaving each state, which equals the exponential rate of leaving such a state: $\beta \rs_{-1} \ri_{0} /N + \gamma \ri_0 + \delta \rr_0$. The decay parameter for the process is therefore the minimal such value across all non-absorbing states. Therefore $\rs_{-1} = 0$, $\ri_0 = 1$ and $\rr_0 = 0$ minimise the decay parameter. According to Theorem 1 of \cite{vanDoorn2008} this forces the QSD to have full mass on this state where $\ri_0 = 1, \rs_0 = N-1$, since the only state accessible from this state is an absorbing one.
\end{proof}

\subsection{Sampling the quasi-stationary distribution}\label{sampling}

Samples from quasi-stationary distributions can be produced using the sequential Monte Carlo (SMC) sampler with stopping time resampling methods developed in \cite{griffin2016}. In brief, $M$ realisations of the model (referred to as particles) are simulated forward in time. Absorbed particles (with no infected individuals) are given weight zero and non-absorbed particles are given weight 1 initially. The distribution of weights converges to the limiting conditional distribution. Once the total weight drops below a proportion $\lambda$ of the initial weight, the particles are replenished via a resampling step. Combine-split resampling \cite{griffin2016} was used, which prevents any occupied states from being lost and has the advantage that the distribution of weights remains unchanged after resampling. This resampling method combines particles in the same location together, draws new particle locations uniformly from the existing locations and equalizes the weight on particles in the same location.  In our implementation, after a burn-in of $T_\text{b}=1$, weighted samples were drawn every $T_\text{d}=1$ time units until time $T_\text{max}=100$. 

Figure \ref{fig: example sirs} shows the expected number of individuals in each class under the QSD. In this example we used $M=1000$ particles and a resampling threshold of $\lambda=0.6$. The Figure shows that when $\delta$ is smaller there is a larger proportion of globally immune individuals in the removed classes, and so the population can support fewer strains.

\begin{figure}[t!]
	\centering
 		\includegraphics{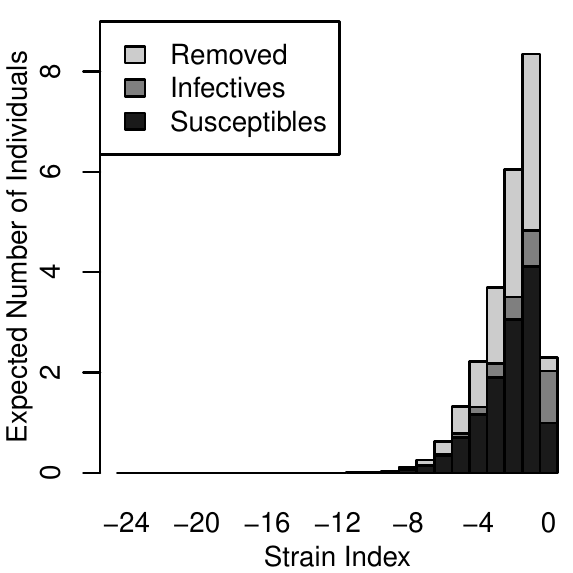}
 		\includegraphics{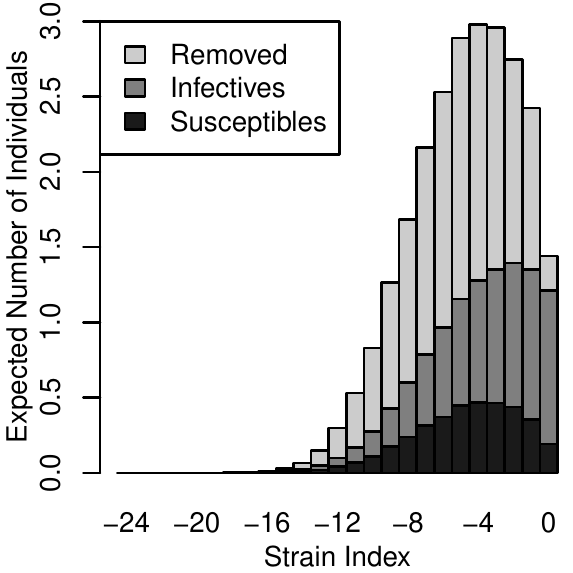}
	\caption{Comparisons of expected population composition under E-SIRS QSDs with (a) $\beta = 2$, $\theta = 0.9$, $\delta=0.2$, $N = 25$; (b) $\beta=2$, $\theta = 0.9$, $\delta=2$, $N=25$.}
	\label{fig: example sirs}
\end{figure}

\section{Limiting behaviour}\label{sec: limit}
One aspect of interest is how the E-SIS and E-SIRS processes relate to those without mutation. To this end, we consider the limits of the times to extinction of the processes as $\theta$ tends to 0 or 1, and the limit, for fixed $\theta$, of the time to extinction as the population size $N$ tends to infinity. Large population limits can be used to justify the use of infinite population models as approximations for, for example, the decay parameters for the relevant processes which we cannot obtain analytically. 

\subsection{Limits as mutation probability changes}\label{subsec: theta limit}

\begin{mythm}\label{thm: tt 1}
	Let $T^{\theta}$ be the time to extinction of the E-SIS model, and $T^1$ the time to extinction of the standard SIS model, both starting from a single infective (nominally of strain index 1) in a population of $N$ individuals. Then $T^{\theta} \rightarrow T^1$ in distribution as $\theta \rightarrow 1$.
\end{mythm}
\begin{proof}
		We make use of a coupling of $(T^{\theta}: 0 <\theta < 1)$ and $T^1$. 
	 Firstly, let $\bX^\theta(t) = (\bS^\theta(t), \bI^\theta(t))$ be the E-SIS model.
	We assume the process to be defined over a population indexed by $n=1, \dots, N$.
	\begin{itemize}
		\item For each individual $n$, define a sequence of i.i.d. infectious periods $\{ L^{(n)}_m \sim \Exp (\gamma): m \in \NN \}$. 
		\item For each \emph{ordered} pair of individuals $(n, n')$, define a homogeneous Poisson process $\{A^{(n,n')}(t)\}$ on $[0,\infty)$ with rate $\beta/N$. 
		\item For each ordered pair $(n, n')$ define a sequence of uniform random variables $U^{(n,n')}_l \sim \text{Unif}[0,1]$ for $l \in \NN$. 
		\item Let all $L^{(n)}_m$, $A^{(n,n')}$ and $U_l^{(n,n')}$ be independent.
	\end{itemize}
	Now let $(E, \cF, \PP)$ be the product space of these random processes and variables.
	Using these building blocks, we construct an E-SIS model $\{\bX^\theta(t)\}$ and SIS model $\{\bY(t) = (S(t), I(t))\}$ as follows. Set $S^\theta_{-1}(0) = N-1, I^\theta_0(0) = 1$ for the E-SIS model and set $S(0) = N-1$, and $I(0) = 1$ for $\bY(0)$. Assume the initial infective individual has index $n=1$ without loss of generality. In both processes infectious individual $n$ generates contacts with each susceptible individual $n'$ at points of the Poisson process $\{A^{(n,n')}(u)\}$, where $u$ denotes the cumulative time that individual $n$ has been infectious and $n'$ has been susceptible. In other words the Poisson processes are stopped whenever it is not possible for individual $n$ to infect individual $n'$. At each contact event in the SIS model an infection occurs. The newly infected individual $n'$ stays infected for a period of length $L^{(n')}_{m(n',t)+1}$, where $m(n',t)$ is the number of infections individual $n'$ has recovered from up to the contact time $t$. In the E-SIS model the $i$th contact event in $\{A^{(n,n')}(u)\}$ results in a mutation if and only if $U^{(n,n')}_i \leq \theta$, in which case individual $n'$ is infected with a new strain and given the lowest unused strain index. However if $U^{(n,n')}_i > \theta$ then individual $n$ attempts to infect $n'$ with their current strain and the infection is successful if the strain index of individual $n'$ is strictly less than the stain index of $n$. As in the SIS model, successful infections in the E-SIS model are given infectious period length $L^{(n')}_{m(n',t)+1}$. Notice that under this coupling non-mutation contacts of $n'$ by $n$ are only successful if the strain index of $n$ is strictly greater than that of $n'$ in the E-SIS model. However mutation contacts and all contacts in the SIS model are always successful.
 
	Fix $\omega \in E$, our probability space. For the SIS model, we have $T^1 < \infty$ almost surely. On the interval $[0, T^1(\omega))$, there are two possibilities for the E-SIS model. At each infective-susceptible contact we compare the strain indices. The first possibility is that every contact will always lead to a successful infection, arising from a sequence of infection events which always contact a susceptible of a lower index. In this case, we have $T^{\theta}(\omega) = T^1(\omega)$ for all $\theta \in [0,1]$. The second possibility is one or more  ``potentially unsuccessful'' contact events exist, in which if the event were to be non-mutation, it would fail. This failure occurs if the relevant $U^{(n,n')}_i > \theta$. Since we must have a finite number of such events occurring in $[0, T^1)$, we can find $\theta^*$ such that $U^{(n,n')}_i \leq \theta^*$ for all such $U^{(n,n')}_i$ corresponding to potentially unsuccessful events. This means that for $\theta\geq \theta^*$ we must have $T^{\theta^*}(\omega) = T^1(\omega)$. So for every $\omega \in E$, there exists $\theta^* \in (0,1)$ such that $T^{\theta}(\omega) = T^1(\omega)$ for all $\theta \geq \theta^*$. Hence $T^{\theta}(\omega) \rightarrow T^1(\omega)$ as $\theta \rightarrow 1$ for almost every $\omega \in E$, and hence $T^\theta\rightarrow T^1$ in distribution by the Skorohod Dudley theorem from, for example, Theorem 3 of \cite{dudley1968}.
\end{proof}

\begin{mythm}\label{thm: tt 0}
	Let $T^0$ be the time to extinction of the standard SIR model. Then $T^{\theta} \rightarrow T^0$ in distribution as $\theta \rightarrow 0$.
\end{mythm}
Intuitively, one can think of identifying the $\cS_1$-class for the E-SIS model 
and the $R$-class of the SIR model. As mutation events get rarer, the chance of mutation happening before extinction becomes smaller and smaller, and so the two processes are more likely to coincide under a suitable coupling until extinction.
\begin{proof}
This proof follows in a similar fashion to Theorem \ref{thm: tt 1}. In this version, the coupling is constructed between the E-SIS and the SIR model. The only differences are that in the SIR model individuals enter the removed category after their infectious period and the Poisson process $\{A^{(n,n')}(u)\}$ progresses during any time for which $n$ is infective and $n'$ is susceptible in the E-SIS model (as before); but when $n'$ is susceptible \emph{or removed} in the SIR model. 

In the E-SIS model infectious contacts between $n$ and $n'$ are only successful if the event is a mutation or $n'$ is of a strictly lower strain index than $n$. In the SIR model, only the first infectious contact is successful. This means that the two epidemics must be identical up to the time of the first repeat contact, when one identifies the $\{\cS_1, \cS_2, \cS_3, \dots\}$ classes in the E-SIS model with the $R$ class of the SIR model.

Similar to the proof of Theorem \ref{thm: tt 1}, for each $\omega\in E$ one can find a value of $\theta^*$ so that $T^{\theta^*}(\omega) = T^0(\omega)$ for all $\theta \leq \theta^*$, and so $T^{\theta} \rightarrow T^0$ almost surely as $\theta\rightarrow 0$ and hence $T^\theta\rightarrow T^1$ in distribution by the Skorohod Dudley theorem of \cite{dudley1968}.
\end{proof}

\subsection{Large population limits}\label{subsec: bdps}
In order to obtain some large population limit results, we will consider an ``infinite'' population model. We will refer to this as an \emph{Evolving Birth-Death Process} (E-BDP). More precisely we assume that infected individuals are a negligible part of an infinite population of individuals that are not immune to any strains at the start of the epidemic, and so all infections will be successful almost surely.  This implies infections from a given strain $k$ and recoveries from that strain behave as a linear birth-death process with birth rate $\beta$ and death rate $\gamma$. Additionally, at the point of each infection, with probability $\theta \in [0,1]$, the new infective is infected with a previously unseen strain, and given the next available strain index $K^* + 1$.

The possible events comprise:
\begin{itemize}
	\item \emph{Infection with mutation:} $\cI_{K^* + 1}(t)=0 \mapsto \cI_{K^* + 1}(t)=1$
	with rate $\beta \theta \cI(t)$.
	\item \emph{Infection without mutation:} $\cI_k(t) \mapsto \cI_k(t) + 1$
	for $k \in \ZZ$ with rate  \mbox{$\beta(1-\theta)\cI_k(t)$}.	
	\item \emph{Recovery:}  $\cI_k(t) \mapsto \cI_k(t) - 1$
	for $k \in \ZZ$ with rate $\gamma \cI_k(t)$.
	
\end{itemize}

After it emerges, each strain behaves according to a linear birth-death process with birth rate $\beta(1-\theta)$ and death rate $\gamma$. The total number of infectives $\cI(t)$ also behaves according a birth-death process with birth rate $\beta$ and death rate $\gamma$.

The time to extinction of the E-SIS model converges to that of the E-BDP model, noting that under a suitable coupling, the time to extinction of the E-BDP equals the Linear BDP without mutation. This leads us to the following result.
\begin{mythm}
	Let $T^{\theta,N}$ be the time to extinction of the E-SIS model, and $T$ the time to extinction of the E-BDP model. Then we have $T^{\theta,N} \rightarrow T$ in distribution as $N \rightarrow \infty$ when $\beta < \gamma$. If $\beta \geq \gamma$, then on the event $\{T < \infty\}$, a region of probability $1 - \gamma/\beta$, we also have $T^{\theta,N} \rightarrow T$ in distribution as $N \rightarrow \infty$
\end{mythm}

\begin{proof}
	Using Theorems \ref{thm: tt 1} and \ref{thm: tt 0} we can conclude that for any fixed $N$ that $T^{0,N}$ is the time to extinction for the standard SIR model, and $T^{1,N}$ is equal to the time to extinction for the standard SIS epidemic model. Furthermore, from these theorems we can construct a coupling of the SIS, E-SIS and SIR models using two sets of Poisson processes and mutation indicator variables such that, for any $\theta \in [0,1]$,
\begin{equation}\label{bounds}
		T^{0,N}(\omega) \leq T^{\theta,N}(\omega) \leq T^{1,N}(\omega)
\end{equation}
for almost every $\omega\in E$. From \cite{barbour1975}, we know that if $\beta < \gamma$ then $T^{0,N}$ converges in distribution to $T$, the time to extinction of a Linear BDP with the same parameters $\beta$ and $\gamma$. From \cite{andersson1998} we obtain that the same thing happens for SIS models, i.e. $T^{1,N} \rightarrow T$ in distribution as $N \rightarrow \infty$. Using the bounds in Equation \eqref{bounds}, we obtain that $T^{\theta,N} \rightarrow T$ as $N \rightarrow \infty$ for all $\theta\in[0,1]$. 
	
	In the case where $\beta \geq \gamma$ we note that on a set of probability $1 - \gamma/\beta$, the time to extinction of the linear BDP is infinite, as discussed in Chapter 3.2 of \cite{Anderson1991}. From \cite{andersson1998}, we know that $T^{1,N} \rightarrow T$ almost surely (and hence in distribution) on the event $\{T < \infty \}$. From \cite{barbour1975} we know that on this event, $T^{0,N} \rightarrow T$ in distribution. Therefore we must have that $T^{\theta,N} \rightarrow T$ as $N \rightarrow \infty$ for all $\theta\in[0,1]$ here too. 
\end{proof}
It should be noted, that on the event $\{T = \infty\}$ we don't have $T^{0,N} \rightarrow \infty$. Instead $T^{0,N}$ converges to an extreme-value distribution as mentioned in Theorem 8.1 \cite{Andersson2000}.

Next we show existence of a QSD for the E-BDP model.

\begin{mythm}\label{thm: eu bdps}
	Let $\bX(t)$ be the E-BDP with parameters $\gamma > \beta > 0$ and $\theta \in [0,1]$. Then, under the equivalence relation described in Section \ref{sec: equiv} and conditional on the event $\{I(t) > 0\}$, there exists a unique QSD.
\end{mythm}
\begin{proof}
	To prove existence, we first show that the state space we are interested in is countable. To do this we use the following construction. Starting with a single infective of strain 0, we can define a method of constructing the state space. By having a birth in strain 0, or a mutation event, one can systematically arrive at any state in the state space. Given these two possible events, one can encode each state according to a finite binary sequence, which corresponds to a unique integer which we can use to enumerate the space. Given that there exists a lower bound $l \in \ZZ$ such that $I_k = 0$ for all $k \leq l$, we construct the state as follows. Starting with the lowest non-zero strain index $l+1$  consider $I_{l+1}$ strain 0 within-strain-infection events. Then for each higher strain $k$, we choose a mutation event followed by $I_k - 1$ within-strain infection events. 
	Note that only considering finite sequences gives countability, unlike the uncountability of the infinite paths on this binary tree.

	To obtain existence of a QSD, we now introduce a coupling. Let $\bX(t) = (X_j(t))_{j \in \ZZ}$ be the E-BDP. Let $\alpha^X$ be the decay parameter for $\bX(t)$. Let $(Y(t))_{t \geq 0}$ be the process defined on the same probability space, given by $Y(t) = \sum_{j \in Z} X_j(t)$. Since the mutations do not affect whether infections are successful or not, $Y(t)$ can be seen to be a single-strain linear BDP with birth rate $\beta$, and death rate $\gamma$. As discussed in Example 1 of \cite{vanDoorn1991}, $Y(t)$ has the decay parameter $\alpha^Y = \gamma - \beta$. Let $T_X$ be the extinction time of $\bX(t)$ and $T_Y$ for $Y(t)$.
	
	Letting $\alpha^Y_0$ be the absorption parameter for $Y(t)$, and $\alpha_0^X$ for $\bX(t)$, we also know that $\alpha_0^Y = \gamma - \beta$. Since $T_X = T_Y$ under the coupling, we use the definition of the decay parameter to deduce that $\alpha_0^X = \gamma - \beta$, and hence $\alpha^X \geq \alpha_0^X > 0$. Using Theorem 13 of \cite{vanDoorn2013} we get existence of a QSD.
	Moreover, using Theorem 3.3.2 of \cite{jacka1995}, we must have $\alpha_X = \gamma - \beta$ since there is only one state from which extinction can occur: one must have 1 infective before extinction, which must be of strain $0$ under the equivalence relation. This leads to the uniqueness of the QSD.
\end{proof}

\section{Reproduction numbers}\label{sec: repro}
To characterise the dynamics of the models, we look to a number of key statistics which are related to the commonly used \emph{basic reproduction number}, $R_0$, that illustrates whether or not an epidemic is likely to infect a large proportion of the population. The basic reproduction number is defined as the number of individuals infected by a single typical infective in a large, otherwise susceptible population \cite{Anderson1992}. In the E-SIRS model, we still have $R_0 = \beta/\gamma$. One issue with $R_0$ is that it fails to take into account the likely immunities present in the population, or how much the pathogen evolves during the opening phase of the epidemic. 

\subsection{Modified household reproduction number $R_*$}\label{subsec: rstar}
In \cite{Ball1997}, an epidemic is considered which spreads through a population grouped into households, such that individuals in the same household make contact at a different rate to individuals in different households. The households reproduction number $R_*$ is shown in \cite[Section 2.3]{Ball1997} to be equal to $R_* = \mu R_H$
where $\mu$ is the expected number of individuals infected in a single household epidemic (including the initial infective), and $R_H$ is the the mean number of contacts an infective individual makes with individuals in other households during a single infectious period. 

For the E-SIRS model, we consider each strain as a ``household'' which has countably many individuals, and mutations are considered contacts between households. In this case $\mu$ is the expected total population of a birth death process with birth rate $\beta(1-\theta)$ and death rate $\gamma$, including the initial infective. One can use the branching property to compute $\EE[Z]=\gamma/(\gamma-\beta(1-\theta))$ and note that the between household reproduction rate $R_H=\beta\theta/\gamma$ and so,
\begin{align*}
R_* = \begin{cases}
\frac{\beta\theta}{\gamma - \beta(1-\theta)} & \beta(1-\theta) < \gamma \\
\infty & \beta(1-\theta) \geq \gamma
\end{cases}
\end{align*}

To recontextualise this in terms of strains and mutations, one can think of $R_H$ as the expected number of new strains originating from a single individual during one infectious period, and $\mu$ as the expected number of individuals that ever get infected by a specific strain.

One could consider $R_0$ to be the ``intra-strain'' reproduction number, and $\mu$ to be the ``inter-strain'' reproduction number. With these we obtain one of three regimes:
\begin{itemize}
	\item If $R_0 = \beta/\gamma < 1$, then the whole population would die out with certainty, and no large epidemic would occur.
	\item If $R_0 \geq 1$ and $\mu < \infty$ then a large epidemic occurs with positive probability, but each individual strain dies out quickly.
	\item If $R_0 \geq 1$ and $\mu = \infty$, then each strain has a positive probability of producing a large outbreak.
\end{itemize}
Figure \ref{fig: tree} shows realisations of the genetic trees under the E-SIS model under the two supercritical regimes. For small $\theta$, we obtain only a small number of strains, and the epidemic is more likely to die out. Moreover, in a finite population, this low $\theta$ leads to high immunity in the population and hence shorter epidemics. The trees highlight how only a small number of the strains survive for a long time, particularly in Fig. \ref{fig: tree}(a). Figure \ref{fig: tree} also show similarities to the tree for H3N2 in Extended Data Figure 9(c) of \cite{bedford2015}, a paper which specifically looks to model influenza.

\begin{figure}[h!]
	\centering
 		\includegraphics[width=1\textwidth]{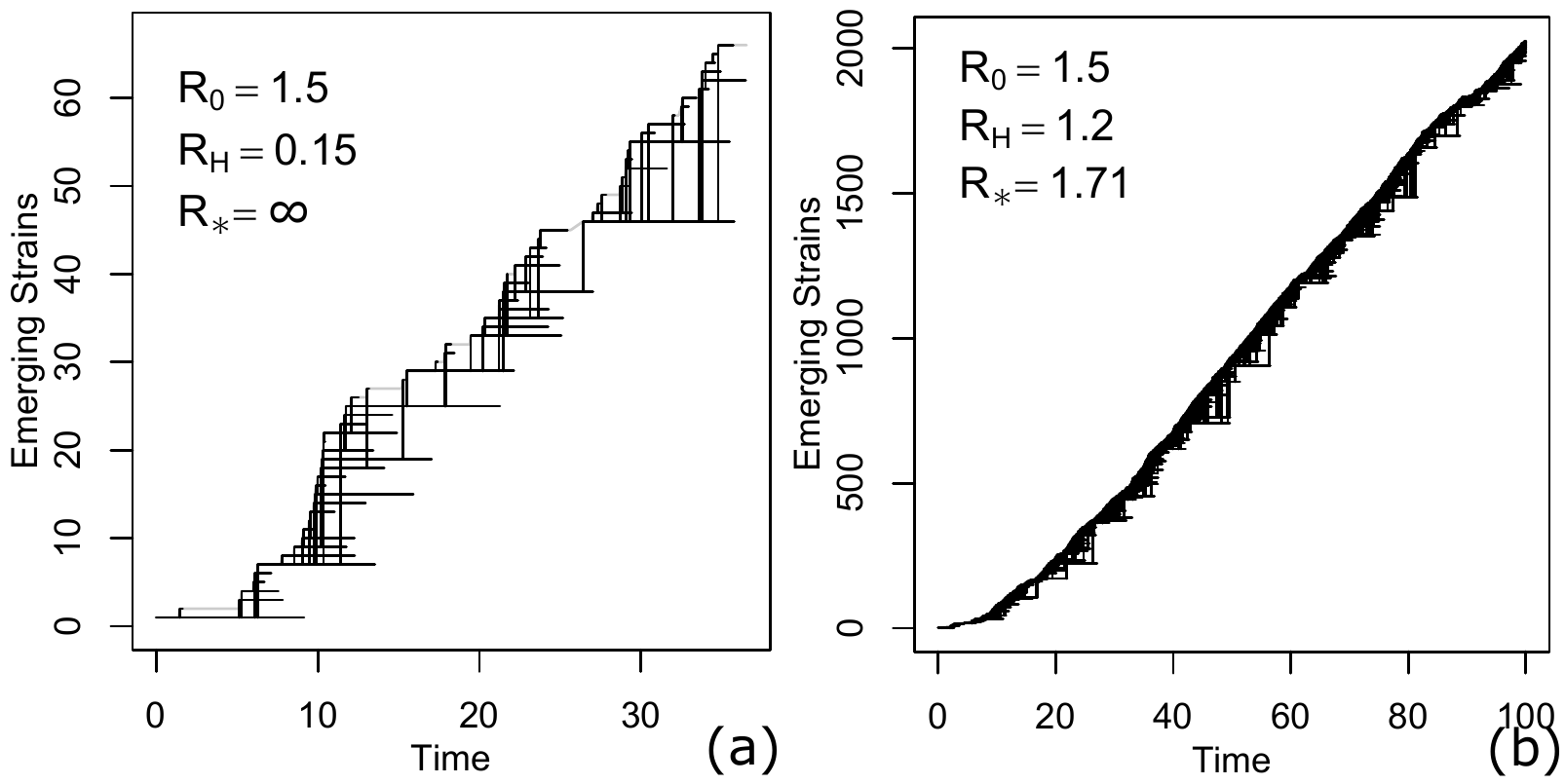}
	\caption{Comparisons of emergence of strains under different $R_0$, $R_H$ under the E-SIS model with $N=25$, $\gamma = 1$, $\beta = 1.5$ with (a) $\theta=0.1$; (b) $\theta=0.8$.}
	\label{fig: tree}
\end{figure}

\subsection{Quasi-stationary reproduction number $R_Q$}\label{subsec: rt}
One drawback to the $R_0$ is that it only usefully describes the initial behaviour of an epidemic in a naive population and doesn't take into account the build up of immunity in the E-SIRS model. One alternative is to consider the \emph{effective reproduction number}, denoted $R_t$, defined as $R_t = R_0 \frac{S(t)}{N}$ in a population of size $N$. Much work has been done in trying to evaluate $R_t$ for specific infections such as influenza by \cite{cowling2010} and Ebola by \cite{althaus2014}. However, $R_t$ is time-dependent and can therefore be difficult to compute and interpret. At a quasi-stable equilibrium the number of new infections balances the recoveries and so $R_t\approx 1$, and hence $R_t$ is not informative about the disease characteristics. Ideally, we would like a reproduction number that adjusts for the build-up of immunity in the population, but remains informative about the infectivity of a disease. 

We offer an alternative reproduction number, based on the QSD, which aims to describe the infectiousness of strains of an endemic disease in a population with `average' levels of historical immunity. The quasi-stationary reproduction number ($R_Q$) is the average number of secondary infections caused by a single typical infective introduced into an otherwise uninfected (S status) population with levels of immunity (strain indexes) drawn from the quasi-stationary distribution, so each other individual may or may not be immune to the current strain of the infective. By typical infective, we mean an individual with strain index sampled from the distribution of strain indexes of infectives in the QSD. Under the E-SIRS model, the total number of infectives is always less than the SIS model without evolving strains, and so $R_Q\leq R_0$.

The quasi-stationary reproduction number provides a measure of the ability of a pathogen to re-invade a population from which it has been eradicated. For diseases like seasonal influenza which have greatly reduced incidence during the summer months, $R_Q$ measures the reproduction number at the beginning of the next influenza season after accounting for the residual immunity left over from last year.

More precisely, we draw the single infective from the marginal number of infectives in the QSD $u_I(k)$: the probability that given an individual is infective, it is of strain index $k$. For QSD $\bu$ this is given by 
\[
	u_I(k) = \sum_{(\bs,\bi,\br) \in \Omega} u_{(\bs,\bi,\br)} \frac{{\rm i}_k }{ |\bi| }.
\]
Under the equivalence relation described in Section \ref{sec: equiv}, we can have a maximum of $N$ strains in a population of size $N$, and so the strain index ranges over $k \in \KK^* = \{0, 1-N\}$. The susceptible population is drawn from the total strain marginals of the QSD $u_K(k)$: the probability that under the QSD that a given individual is of strain index $k$. 
\[
	u_K(k) = \sum_{(\bs,\bi,\br) \in \Omega} u_{(\bs,\bi,\br)} \frac{\textrm{i}_k + \textrm{r}_k + \textrm{s}_k}{N} 
\]
Finally, we require the probability that a randomly chosen individual drawn from the strain marginal will be susceptible to strain $k$ (i.e.\ will have a strain index lower than $k$):
\[u_L(k) = \sum_{j=1-N}^{k-1}u_K(j).\] 
During their infectious period the infective makes infectious contact with each individual at the points of a Poisson process with rate $\beta/N$. For large populations the infective is unlikely to contact the same individual twice (or themselves), and so the expected number of contacts is $\beta/\gamma$. With probability $\theta$ the contacts are mutations and are successful infections. With probability $(1-\theta)$ the contacts are non-mutations and are only successful if the individual contacted has a lower strain index, which occurs with probability $u_L(k)$ when the infective has strain index $k$. To calculate $R_Q$ we condition on the strain index of the initial infective, hence
\begin{align}\label{eqn: rq}
R_Q&=\frac{\beta}{\gamma}\sum\limits_{k=1-N}^{0}(\theta + (1-\theta)u_L(k))u_I(k)\nonumber\\
&=\frac{\beta}{\gamma}(\theta+(1-\theta)\bu_L^T\bu_I).
\end{align}  
Since $u_I$ and $u_L$ are both probability mass functions, $0\leq\bu_L^T\bu_I\leq1$ and so we have that $R_H\leq R_Q\leq R_0$. As $\theta \rightarrow 1$ then $R_Q\rightarrow\beta/\gamma=R_0$, as does $R_H$. 

The three notions of a reproduction number in this section describe three different facets of the epidemic model, and can be compared in Figure \ref{fig: ro comp}. It shows that $R_Q$ is always less than $R_0$ due to the effects of immunity, and $R_*$ depends greatly on $\theta$; the unplotted points for $R_*$ are in the regions where it is infinite, namely where $\beta(1-\theta) \geq \gamma$. Values of $R_Q$ were calculated using the SMC sampler described in Section \ref{sampling} with $M=100$ particles and a resampling threshold of $\lambda=0.4$.

\begin{figure}
	\centering
 		\includegraphics[width=1\textwidth]{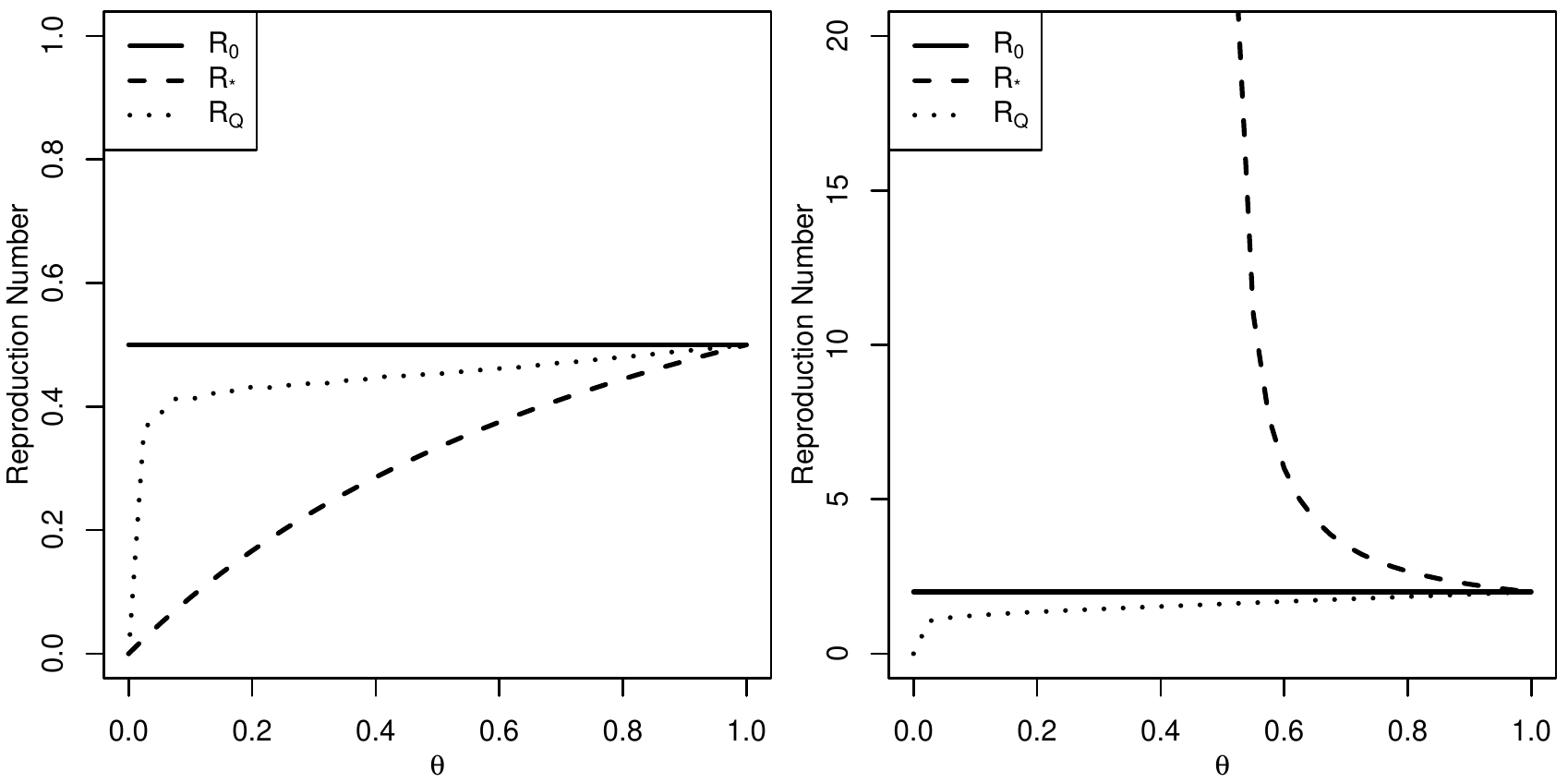}
	\caption{Comparisons of $R_0$, $R_*$ and $R_Q$ under varying $\theta$. with (a) $\beta = 0.5$, $\gamma = 1$, $N = 100$; (b) $\beta=2$, $\gamma=1$, $N=100$.}
	\label{fig: ro comp}
\end{figure}

\section{Simulation Study}\label{sec: sim study}
To further explore the E-SIRS model we use the SMC sampler described in Section \ref{sampling} to investigate numerically features of the QSD which we cannot obtain analytically. We wish to observe how various key properties behave as we vary parameters of the model. To this end we look at the following expectations over the QSD. For brevity we omit time indices and conditioning, and denote expectations under the QSD by $\EE_Q$.
\begin{itemize}
	\item The expected total number of infectives $\EE_Q[I]$ and immune individuals $\EE_Q[R]$ in the QSD, where $I = \sum_{k=-\infty}^0 I_k$, $R = \sum_{k=-\infty}^0 R_k$. 
	\item The expected total number of active strains $\EE_Q[K]$ in the QSD where $K = |\{k : I_k > 0\}|$. 
	\item We also look at how varying the model parameters affects strain diversity in infectives and the whole population.
\end{itemize}
We will focus on the E-SIS model, but also discuss for each statistic how the addition of an immune period, as in the E-SIRS models, changes the number of infectives and strain diversity. Unless otherwise stated, all expectations over the QSD were produced with the SMC sampler described in \ref{sampling} with $M=100$ particles and resampling threshold $\lambda=0.4$.

\subsection{Expected number of infectives}
Figure \ref{fig: bt sis} shows a heatmap of the expected number of infectives in the population under quasi-stationarity, $\EE_Q[I]$, and how this depends on the contact rate and the probability of mutation. Increasing the contact rate $\beta$ or mutation probability $\theta$ increases the expected number of infectives. However, for a fixed population size (in this case $N=100$), the number of infectives increases linearly in $\beta$ when $\EE_Q[I]$ is much smaller than $N$. This can be observed in Figure \ref{fig: b sis}(a), which shows that the number of infectives grows more slowly as $\beta$ increases, especially when $\theta$ is small and so the probability of failed infections is high. As was also noted in Section \ref{sampling}, increased levels of global immunity result in fewer infectives under quasi-stationarity, due to the increased possibility of failed infections.

\begin{figure}
	\centering
	\includegraphics[width=0.7\textwidth]{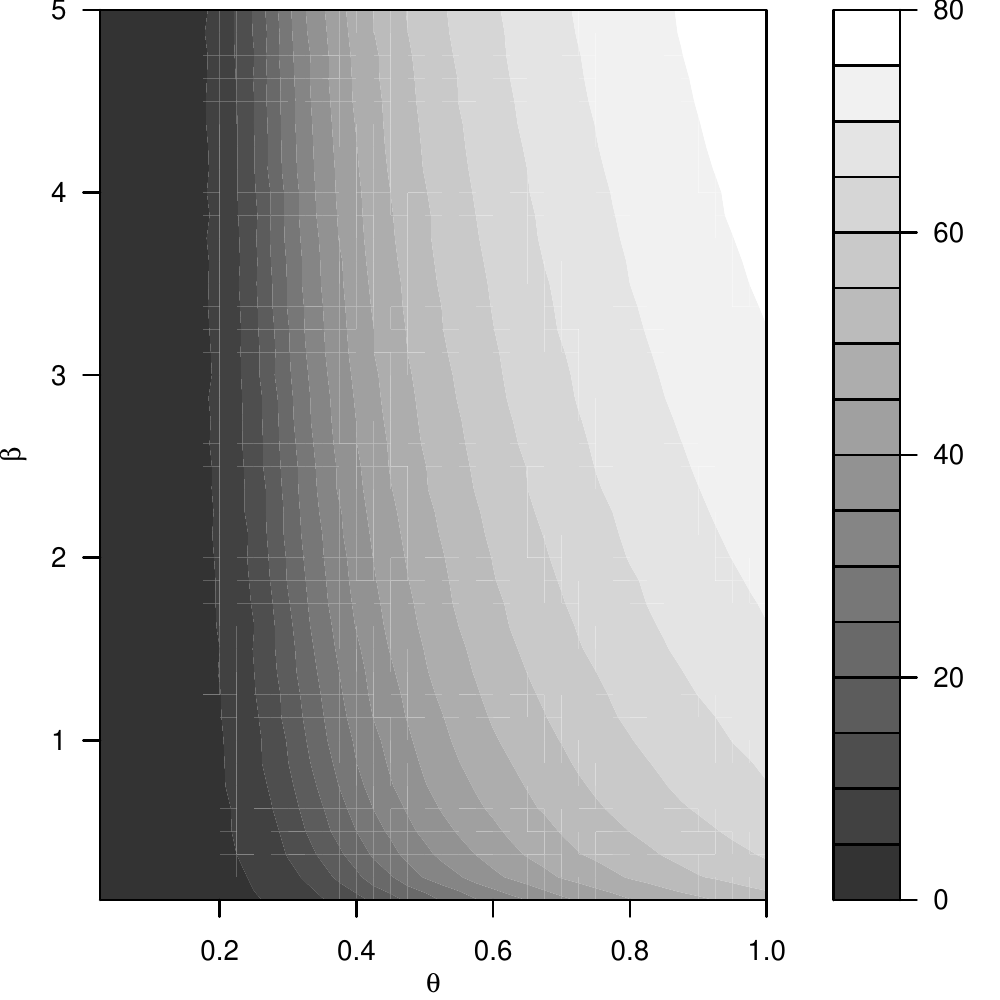}
	\caption[Infectives as $\beta, \theta$ vary in the E-SIS Model.]{Expected number of infectives as $\beta$ and $\theta$ change in E-SIS with $\gamma = 1$, $N=100$.
}
	\label{fig: bt sis}
\end{figure}

Figure \ref{fig: b sis}(b) shows that as $N$ increases the expected proportion of infectives ($\EE_Q[I]/N$) decreases in the case where $\beta < \gamma$, whereas in the supercritical case we see that $\EE_Q[I]/N$ remains fairly constant. In the E-SIRS model, we see that $\EE_Q[I]/N$ is decreased by the introduction of transient global immunity. Furthermore, as $\delta$ gets smaller the transient immunity lasts longer and $\EE_Q[I]/N$ further decreases. 

\begin{figure}
	\centering
 		\mbox{\includegraphics[width=0.5\textwidth]{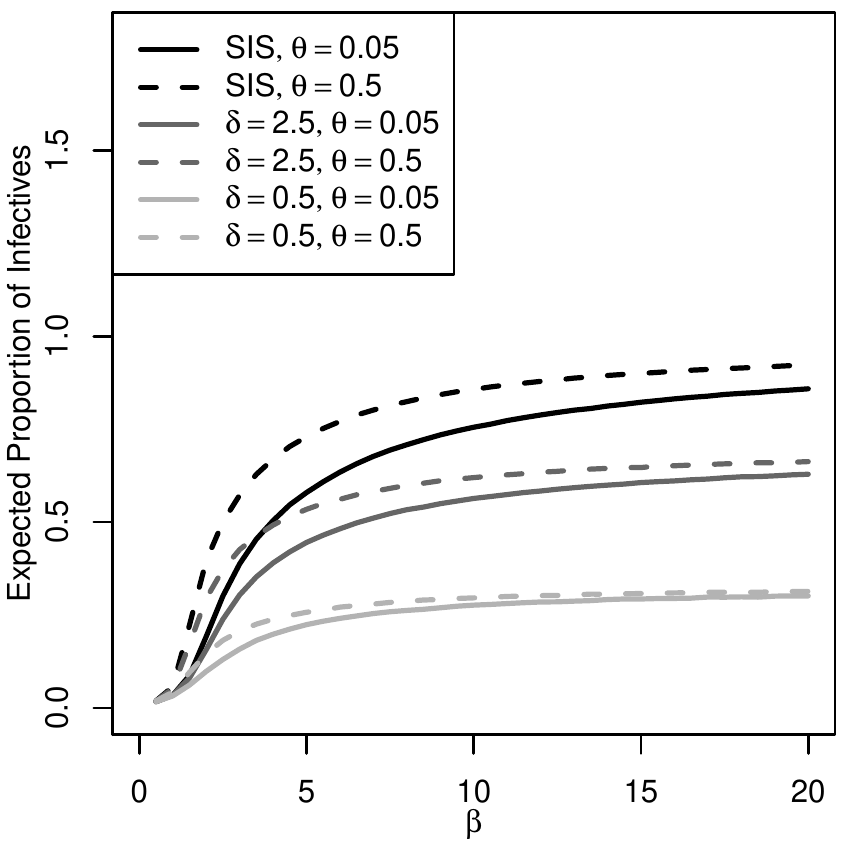}
 		\includegraphics[width=0.5\textwidth]{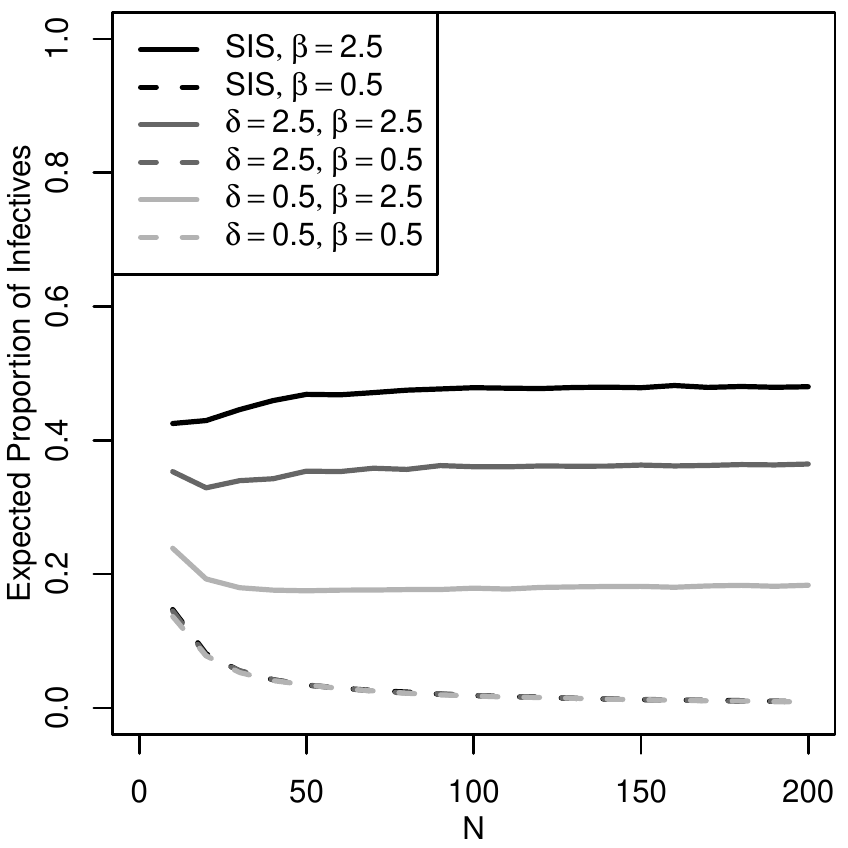}}
	\caption{Expected proportion of infectives in the E-SIS and E-SIRS models: (a) as $\beta$ varies with $\gamma=1, N=100$; (b) as $N$ varies with $\gamma = 1$, $\theta=0.4$.}
	\label{fig: b sis}
\end{figure}

\subsection{Expected number of strains}\label{subsec: strains}
We investigated what happens to the expected number of active strains, $\EE_Q[K]$ (strains held by infectives) as the parameters change. Under our models, the number of strains is always less than the number of infectives due to the absence of super-infectivity (infection of an individual by multiple strains during a single infectious period). As such, much of the behaviour is similar to that of the expected number of infectives in the previous subsection. For example, the expected number of strains increases linearly with $\beta$ when $\EE_Q[K]$ is much less than $N$. This follows since we already know that for $\theta = 1$ every infective begins a new strain and so $\EE_Q[I] = \EE_Q[K]$. At the other end of the scale, we automatically have that $\EE_Q[K] = 1$ if $\theta = 0$.

Figure \ref{fig: fix bt} shows the expected number of strains for fixed $\beta\theta$ (mutation contact rate) and $\beta(1-\theta)$ (non-mutation contact rate), as $\beta$ and $\theta$ vary.  Note that for Figure \ref{fig: fix bt}(b), both $\theta$ and $\beta$ increase from left to right, whereas, to maintain fixed $\beta\theta$, $\beta$ decreases as $\theta$ increases. In the case when $\beta\theta$ is high, one might expect $\EE_Q[I]$ and $\EE_Q[K]$ to be closer in value since there is a high probability of mutation leading to a high number of co-circulating strains. This is demonstrated in Figure \ref{fig: fix bt}(a), where we see that for fixed $\beta\theta$, the number of strains is larger (and therefore closer to the number of infectives) for the $\beta\theta = 2$ line than for the $\beta\theta = 0.05$ line. In Figure \ref{fig: fix bt}(a) there is a maximum point for the number of strains as $\beta$ increases, after which the number of strains decreases. As $\beta(1-\theta)$ increases in Figure \ref{fig: fix bt}(b), the number of strains becomes more linear in $\theta$. 


\begin{figure}[h!]
	\centering
  		\includegraphics{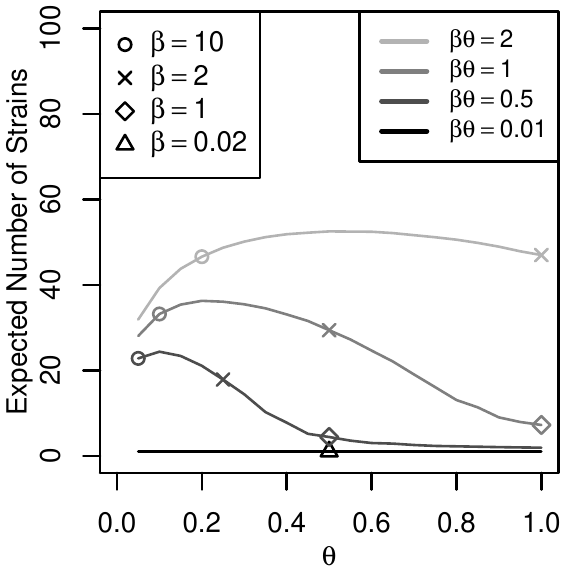}
  		\includegraphics{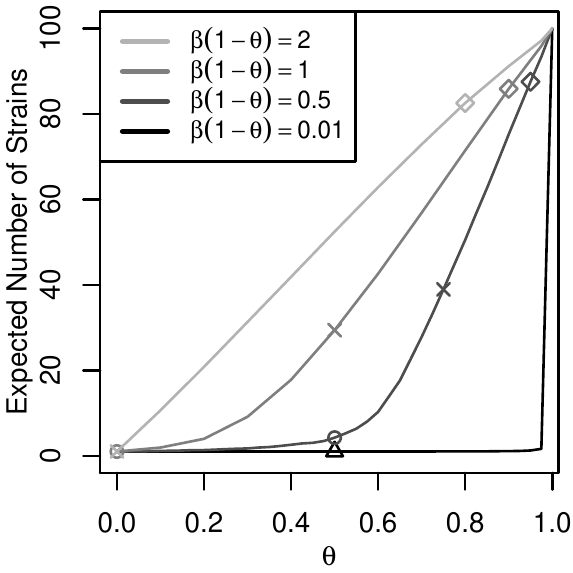}
	\caption{Expected number of strains under quasi-stationarity for (a) fixed $\beta\theta$; (b) fixed $\beta(1-\theta)$; with $\gamma = 1$ and $N=100$. 
	}
	\label{fig: fix bt}
\end{figure}

\subsection{Strain diversity}

In Figure \ref{fig: div} we investigate the distribution of immunity across the active strains.  The figure shows the expected proportion of infectives $\EE_Q[I/N]$ and total individuals for each strain index $\EE_Q[I_k + S_k + R_k]$, relative to the most recently emerged strain. The expectations taken over the quasi-stationary distribution were calculated using the SMC sampler described in Section \ref{sampling}, with $M=100$ particles, resampling threshold $\lambda=0.4$ and burn-in $T_\text{b}=11$.

\begin{figure}
	\centering
 		\mbox{\includegraphics[width=0.5\textwidth]{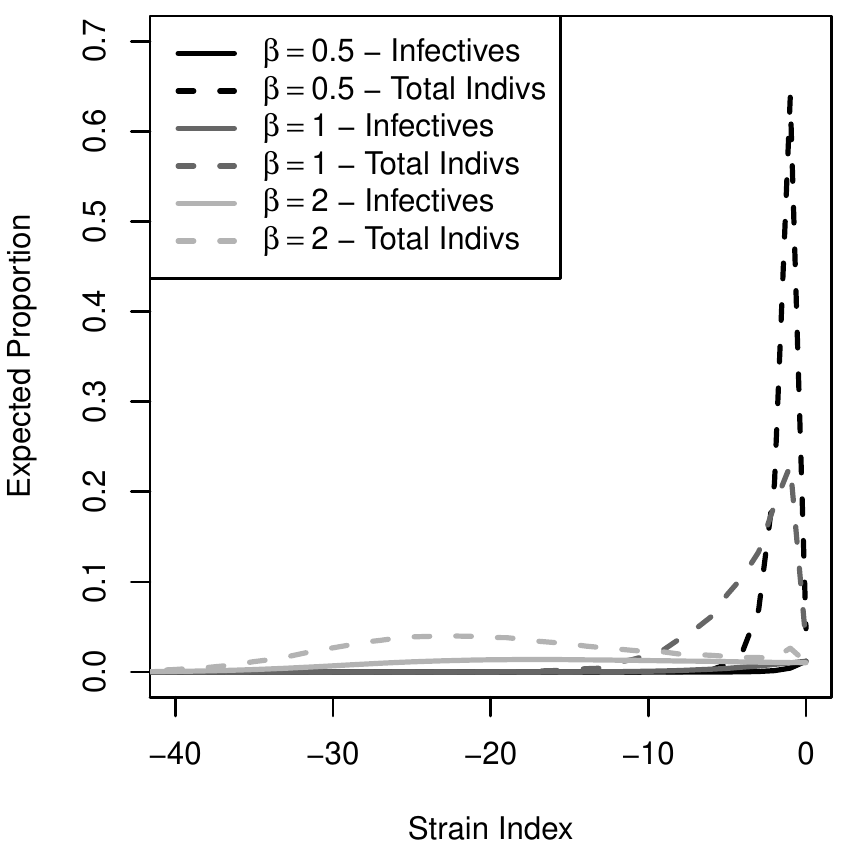}
  		\includegraphics[width=0.5\textwidth]{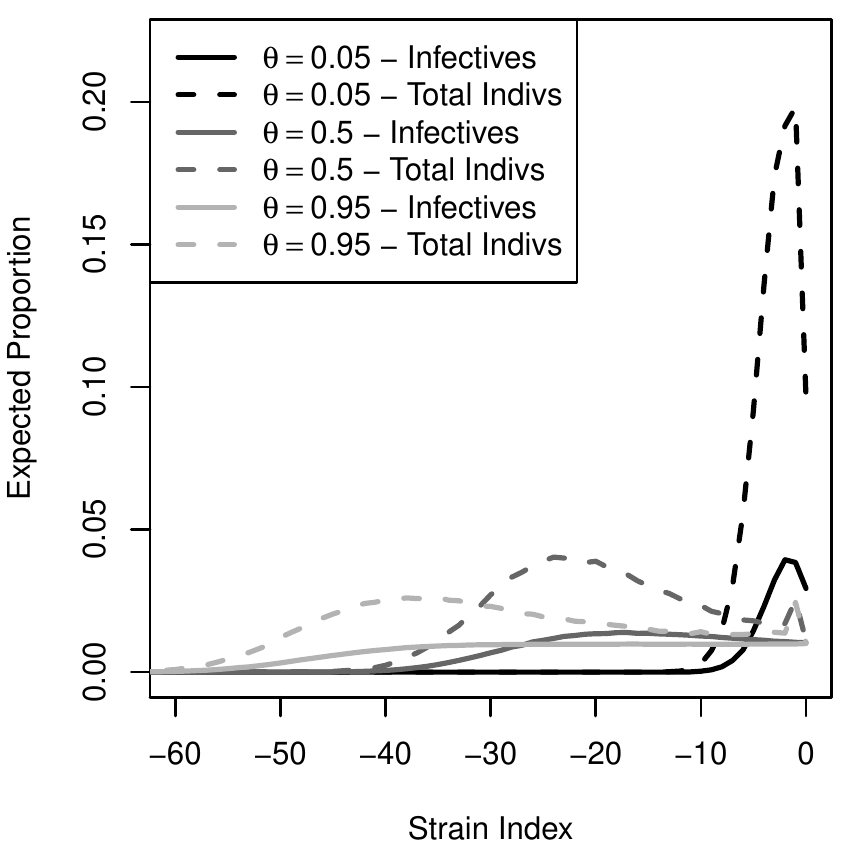}}
  		\mbox{\includegraphics[width=0.5\textwidth]{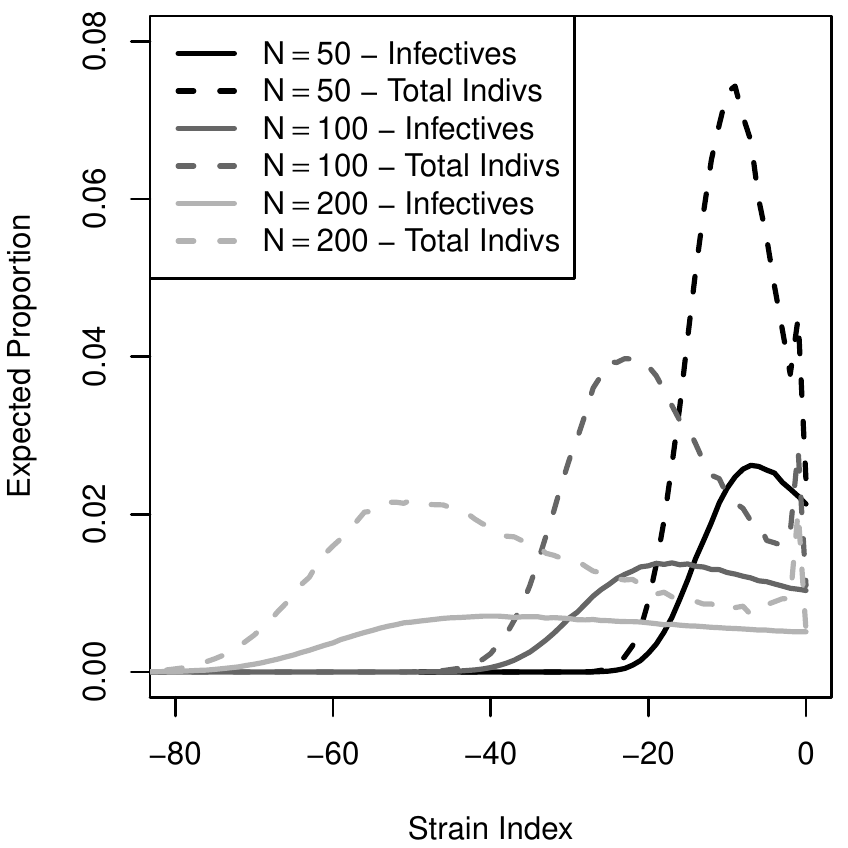}
  		\includegraphics[width=0.5\textwidth]{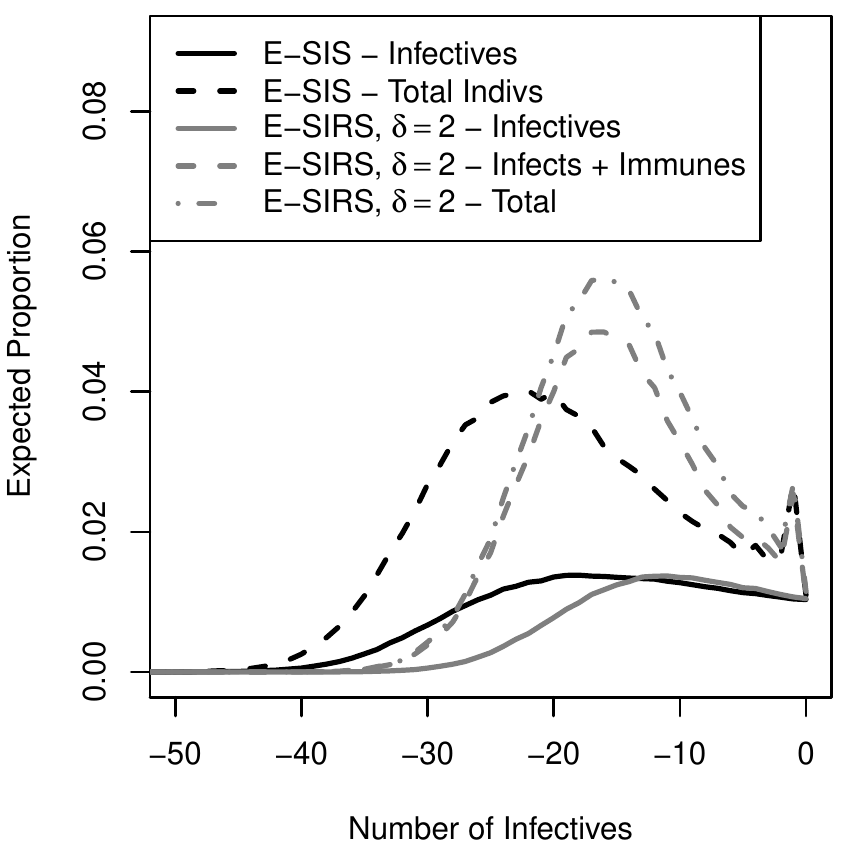}}
	\caption{Strain diversity as a function of the parameters in the E-SIS and E-SIRS models: (a) $\beta$ varies; (b) $\theta$ varies; (c) $N$ varies; (d) $\delta$ varies. Unless otherwise stated $\beta=2$, $\gamma=1$, $N=100$, $\theta=0.5$.}
	\label{fig: div}
\end{figure}

Figure \ref{fig: div}(a) illustrates that larger values of $\beta$ greatly increase the diversity of strains in the infectives and the variation in immunity in the population, since there are more infectives and so more chances for mutation contacts. Another point of interest is the lag of the strain diversity: the number of strains between the mode of the infective strains and the mode of the total population. The lag is fairly consistent for the different values of $\beta$, but does increase slowly in $\beta$. Figure \ref{fig: div}(b) shows the change in strain diversity in $\theta$. As $\theta$ increases, the number of strains present increases, so the strain diversity curve flattens out. For high values of $\theta$, a larger lag is observed between the infectives and the whole population, due to the higher diversity. In Figure \ref{fig: div}(c), the effect of population size is explored. As $N$ increases, we observe a wider number of strains, as one would expect given $\EE_Q[K]$'s behaviour. However, unlike the behaviour as $\beta$ changes, the peak moves away from 0 but the lag between the infectives and the rest of the population appears more consistent. For the E-SIRS model explored in Figure \ref{fig: div}(d), the immune period reduces strain diversity by reducing the expected number of infectives.

In applications, one might wish to look further into the lag between the strain distribution of the infectives and the immunity in the population. If a pathogen has a long lag, then vaccination can be effective in updating the immunity present in the population. However, if the lag is short, then a vaccine based on a recent strain will have little effect in increasing the levels of immunity in the population, as the most immunity profiles in the population will already represent the currently circulating pathogen.

\section{Conclusions}
In this paper we defined an epidemic model for a pathogen undergoing genetic drift, that lies between the well-studied SIR and SIS epidemic models. The model appears to capture some qualitative aspects of the evolution of strains of influenza A, despite depending on just 4 parameters. Compared to models used by \cite{bedford2015} and \cite{parisi2013}, which require the storage of a whole antigenic history, our model is much simpler, which makes simulation, computation and inference much easier. Despite these simplifications, the simulated genetic trees in Figure \ref{fig: tree} show similarities to the tree for H3N2 in Extended Data Figure 9(c) of \cite{bedford2015}. 
The relative simplicity of our model enables analytical insights into model behaviour, such as the relationship between our models and the SIS and SIR models discussed in Theorems  \ref{thm: tt 0} and \ref{thm: tt 1}. A simulation study showed that there is a nonlinear tradeoff between mutation and infectivity when trying to estimate the number of co-circulating strains under quasistationarity. The development of a quasistationary reproduction number $R_Q$ also allows summary the expected behaviour of an epidemic under quasistationarity, by comparing it to a household epidemic model.
Clearly this work could be reduced to a finite state space of strains, but also include more complex strain evolution models, accounting for similarity of strains conferring some amount of partial immunity.

\section*{Acknowledgements}
Funding: AG was supported by EPSRC Grant Number EP/HO23364/1; GOR and SEFS were supported by EPSRC grant EP/R018561/1; SEFS was supported by MRC grant MR/P026400/1. 

\bibliography{biblio}

\end{document}